\documentclass[12pt,Conference,onecolumn]{IEEEtran}

\makeatletter
\def\ps@headings{%
\def\@oddhead{\mbox{}\scriptsize\rightmark \hfil \thepage}%
\def\@evenhead{\scriptsize\thepage \hfil \leftmark\mbox{}}%
\def\@oddfoot{}%
\def\@evenfoot{}}
\makeatother 
\usepackage{graphicx}
\usepackage{epsfig}
\usepackage{setspace}
\usepackage[cmex10]{amsmath}
\usepackage{amsthm}
\usepackage{mathrsfs}
\usepackage{cite}
\usepackage{algorithm}
\usepackage{algorithmic}

\usepackage{amsfonts}
\usepackage{amssymb}
\usepackage{amsmath}
\usepackage{bm}
\usepackage[tight,footnotesize]{subfigure}
\usepackage{subfigure}

\newcommand{\tab}{\hspace*{2em}}
\newtheorem{mydef}{Definition}
\newtheorem{theorem}{Theorem}

\begin{document}

\title{Centralized and Cooperative Transmission of Secure Multiple Unicasts using Network Coding}
\author{\IEEEauthorblockN{Shahriar Etemadi Tajbakhsh \tab Parastoo Sadeghi \tab Rodney Kennedy\\}
\IEEEauthorblockA{Research School of Engineering\\
College of Engineering and Computer Science\\
The Australian National University\\
Canberra, Australia\\
Emails: \texttt{\{shahriar.etemadi-tajbakhsh,
parastoo.sadeghi,rodney.kennedy\}@anu.edu.au}}}

\maketitle
\begin{abstract}
We introduce a method for securely delivering a set of messages to a
group of clients over a broadcast erasure channel where each client
is interested in a distinct message. Each client is able to obtain
its own message but not the others'. In the proposed method the
messages are combined together using a special variant of random
linear network coding. Each client is provided with a private set of
decoding coefficients to decode its own message. Our method provides
security for the transmission sessions against computational
brute-force attacks and also weakly security in information
theoretic sense. As the broadcast channel is assumed to be
erroneous, the missing coded packets should be recovered in some
way. We consider two different scenarios. In the first scenario the
missing packets are retransmitted by the base station (centralized).
In the second scenario the clients cooperate with each other by
exchanging packets (decentralized). In both scenarios, network
coding techniques are exploited to increase the total throughput.
For the case of centralized retransmissions we provide an analytical
approximation for the throughput performance of instantly decodable
network coded (IDNC) retransmissions as well as numerical
experiments. For the decentralized scenario, we propose a new IDNC
based retransmission method where its performance is evaluated via
simulations and analytical approximation. Application of this method
is not limited to our special problem and can be generalized to a
new class of problems introduced in this paper as the cooperative
index coding problem.

\end{abstract}
 \setstretch{2}
\section{Introduction}
A large volume of traffic in data communication networks is dedicated to serving the
demands of individuals or the so-called \emph{unicasts}. With the
rapid growth of the number of network users and their appetite for reliable high data rate multimedia applications, maintaining a desirable quality of service requires
careful and innovative design at different layers of the communication protocol. This is especially true for wireless communication networks, given their scarce
bandwidth resources, and for unicast traffic, given that users
are essentially competing with each other for more bandwidth.

Traditionally, unicasts have been considered as independent flows of
information which should be directed separately toward the intended
destinations. The pioneering work by Ahlswede et al.
\cite{Flow2000}, however, has changed this rigid view of the
network. Originally, \cite{Flow2000} showed that by allowing
different information flows to be combined at intermediate nodes,
the capacity of multicast networks can be achieved. This concept,
which is now known as network coding, has since been studied in
other communication scenarios and network settings including
multiple unicasts \cite{LunUnicast, COPE, MultipleUnicastHARQ} and
broadcast
\cite{BroadcastNguyen,OnlineBroadcast,EryilmazOzdaglar,SamehDelay,
FeedbackAdaptive, ARQMedard} in wireless networks that are subject
to packet erasures. It is noted that optimization of the network
 coded schemes to achieve the best throughput or delay performance in such networks is, in general, a highly non-trivial problem \cite{LunUnicast,OnlineBroadcast}
 and is the subject of on-going research \cite{FisherRecent,SorourRecent,AmyWCNC}. An additional important issue in the case of multiple unicasts is that by mixing different flows and broadcasting them
 over the wireless channel, the secrecy of users' traffic should not be compromised.

The main aim of this paper, which is an extension of the previous
  work by the authors \cite{ITW2012}, is to answer this question: \emph{How can
  one exploit the benefits of network coding for serving multiple unicasts over an
  erasure broadcast wireless channel while maintaining bandwidth efficiency and users' information secrecy?}. To this end, we consider a wireless communication system consisting of one base station
  and a number of users or clients. Initially, the base station uses network coding to combine the information of different users together (in the form of linearly coded packets) and
  broadcasts them to all the users. For now assume that users' information secrecy is somehow ensured. The coded packets are subject to erasures in the wireless channel. Then we study two different
  settings in the system. In the first setting, which we will refer to as the \emph{centralized case}, we assume that the only
  possible way of communication is between each user and the base
  station (the users cannot exchange any information with
  each other). Therefore, the base station will be in charge for retransmission of the missing
  packets until each client can decode and obtain its own information. In the second setting, which we will refer to as the \emph{decentralized or cooperative case},
  after the initial connection to the base station, users who are in vicinity of each other are allowed to
  cooperate with each other (by exchanging information) to compensate for the missing packets and eventually obtain their intended information. The main advantage of the decentralized scenario is that the short range
links among nearby users are faster, cheaper and more reliable.
Moreover a considerable fraction of the bandwidth of the base
station, which was supposed to be dedicated to retransmissions, is
freed for other purposes. The ultimate goal in both settings is to
complete the unicast sessions using as small number of
retransmissions as possible while maintaining users' information
secrecy with proven computational cost of eavesdropping.

%In order to achieve secrecy, the
%  coefficients necessary for decoding packets are kept hidden from everyone but the intended user, who is provided
%  with a private set of decoding coefficients (as the key) to obtain its own information from the received
%  coded packets.

To clarify the problem consider the following example where
   four wireless clients $c_j$, $1\le j \le 4$ are each interested in
  downloading a distinct message $x_j$. Each message is assumed
  to be an element of a finite field $\mathbb{F}_{q}$. The base
  station linearly combines these messages to generate $p_1$, $p_2$, $p_3$ and $p_4$ and broadcasts them to all clients.
  The relationship between the original messages $x_1$ to $x_4$ and the coded ones $p_1$ to $p_4$ is captured through
  \[
 \begin{pmatrix}
  x_1\\
  x_2 \\
  x_3 \\
  x_4
 \end{pmatrix} =
 \begin{pmatrix}
  \alpha_{11} & \alpha_{12} & \alpha_{13} & \alpha_{14} \\
  \alpha_{21} & \alpha_{22} & \alpha_{23} & \alpha_{24} \\
  \alpha_{31} & \alpha_{32} & \alpha_{33} & \alpha_{34} \\
  \alpha_{11} & \alpha_{12} & \alpha_{13} & \alpha_{14}
 \end{pmatrix}\begin{pmatrix}
  p_1 \\
  p_2\\
  p_3 \\
  p_4
 \end{pmatrix}
\]
where $\alpha_{ij}\in \mathbb{F}_q$ are called the decoding
coefficients. However to maintain users' message secrecy, each row of the
decoding coefficient matrix is exclusively made known to the
corresponding user using a combination of public and private channels. For instance, $\alpha_{41}$, $\alpha_{42}$,
$\alpha_{43}$, $\alpha_{44}$ are only provided to $c_4$ by the base
station.

%Client $c_4$ is able to decode its own message
%$x_4=\alpha_{41}p_1+\alpha_{42}p_2+\alpha_{43}p_3+\alpha_{44}p_4$,
%but it cannot decode any other user's message.

Now imagine that each client has received a subset of $P=
\{p_1,p_2,p_3, p_4\}$ due to downlink erasures. As an example,
suppose that the users have initially received
$P_1=\{p_1,p_2,p_3\}$, $P_2=\{p_2,p_3,p_4\}$, $P_3=\{p_3,p_4,p_1\}$,
and $P_4=\{p_4,p_1,p_2\}$, respectively. Also assume that
$\alpha_{jj} = 0$ but $\alpha_{ij} \neq 0$ for $i \neq j$, meaning
that each client $c_j$ does not need $p_j$ but needs all other
packets for decoding its message. So it is clear that each client
has received an unnecessary packet $p_j$ and is still missing one
needed packet. Without network coding, four separate retransmissions
are required to complete each clients' collection. If network coding
is used in the centralized setting, only one transmission by the
base station such as $p_1+p_2+p_3+p_4$ is enough, which is decodable
by all the clients. In the cooperative setting, since none of the
clients has the entire set of coded messages, the total number of
exchanged packets might exceed that in the centralized case. In this
particular example, any client is able to satisfy the demands of the
other three. For example, client $c_2$ can transmit $p_1+p_3+p_4$.
But clearly another transmission is needed to deliver $p_1$ to
$c_2$. Therefore, two transmissions are needed in total, but as
mentioned earlier, these transmissions are generally faster and more
reliable.

\subsection{Our Contributions and Distinctions with Related Work}
We propose a novel method to ensure the secrecy of the users' network decoding coefficients. The core idea is for the base station to privately distribute two permutation functions to each user
 which are then employed to decrypt the location and value  of individual decoding coefficients from two commonly-available location and value sets. We provide the computational cost for an eavesdropper
  for deciphering each user's message. We highlight that there is a key difference between the way we ensure message secrecy using network coding and those such as \cite{SecureNetCod, WeaklySecure, FreeCipher}.
  In particular, the common approach in the literature is to assume
that the eavesdropper becomes aware of the \emph{coding} coefficients along
the wiretapped links because they are publicly broadcast. In this paper, coding coefficients are not broadcast along with the coded packets. Instead, we provide a technique so that the \emph{decoding} (and not the coding)
 coefficients corresponding to each client are privately
delivered to each client. Therefore, an eavesdropper with finite
computational resources would not be able to decode any message even
if it has received all the coded packets. \textbf{Recently we have
shown that our proposed method is provably weakly secure in an
information theoretic sense \cite{ISIT2013}, i.e. the eavesdropper
even with an unbounded computational power can not obtain any
\emph{meaningful} information. The idea of hiding network coding
coefficients has been studied in \cite{CryptAnalysisNetCod} for
multicast scenarios and \cite{MultiResolution} for transmitting
different layers of a multimedia file where each user has access to
a certain number of layers according to its subscription level. Our
scheme considers multiple unicasts and also it is easily extensible
to more general scenarios as it will be discussed in Section
\ref{sec:Broadcast}.}

The centralized retransmission scenario from the base station to the
clients is essentially an index coding problem \cite{IndexCodMain}.
\footnote{\textbf{In the index coding problem, each client has
received a subset of packets $P_j \subset P$ which excludes $p_j$
and requires the single packet $p_j$.
 The goal of the base station who knows all the packets in $P$ is to satisfy the demand of all clients with the minimum number of transmissions using network coding.
 In our problem, if a client is missing multiple packets, it can be broken into multiple virtual clients each with a single missing packet and the index coding problem applies.}}
 Our contribution in this case is to provide an iterative expression to approximate the number of retransmissions from the base station that is required to satisfy the demands
 of all clients in the presence of packet erasures during retransmissions. Our approach builds upon graph representation of the index coding problem in \cite{IndexCodingHeuristics} and the heuristic algorithms in \cite{SamehDelay}
 to find maximal cliques in the graph \textbf{and removing them from the graph by transmitting an instantly decodable combination of corresponding packets}. More specifically, the analysis is based on \textbf{approximating the typical size of the maximal clique found over a
  random graph with appropriately chosen parameters which are updated after each clique removal at each transmission round} . We will verify the theoretical results via simulations.

The decentralized retransmission scenario, in which the clients
cooperate to obtain their missing coded packets, is coined as
cooperative date exchange problem in the literature \cite{ITW2010,
ISIT2010, NetCod2011} and had been also studied in \cite{Repair2008,
RepairNetcod}. To the best of our knowledge, however, all previous
work on the subject have considered the broadcast flow, in which all
clients are interested in all messages.

Our work is more general because it considers multiple unicast flows
with the important notion of secrecy and includes the broadcast
setting as a special case
 (in which case all decoding coefficients are announced publicly to all clients). Our contribution is to first propose \textbf{a random graph-based model for the decentralized setting} and a variant of the
heuristic algorithm in \cite{SamehDelay} for finding maximal cliques
at the clients. \textbf{In the suggested graph model, a local graph
is associated to each client as well as a global graph which is
similar to the graph in \cite{SamehDelay}. Each local graph is
produced according to the set of packets which are held by that
specific client to make sure that client is able to map cliques
removed from its local graph to an instantly decodable packet. The
largest clique over all the local graphs is chosen for removal. All
the other local graphs and the global graph are updated accordingly
upon decoding and obtaining a new packet}. We will verify the
theoretical results via simulations.

Before concluding this section, we highlight some particular
differences with \cite{MultipleUnicastHARQ,weakCCDE} and our
preliminary work in \cite{ITW2012}. The work in
\cite{MultipleUnicastHARQ} provided analytical expressions to
approximate the throughput of network coded hybrid ARQ systems for
multiple unicasts. However, they did not consider the issue of
secrecy. The work in \cite{weakCCDE} considered the problem of
weakly secure coded cooperative data exchange. However, they
considered a broadcast (and not multiple unicasts) setting.
\textbf{In other words, all the target recipients of the entire set
of packets in \cite{weakCCDE}, can decode and obtain all the packets
but the packets remain hidden from an eavesdropper wiretapping some
of the links, however in our work different clients can decode their
own distinct message but the message remains hidden from other
clients and also any external eavesdropper.} \textbf{Finally,the
current paperextends our previous work \cite{ITW2012} in different
ways. \cite{ITW2012} only considers the decentralized scenario where
in the current paper, centralized scenario is added as well.
Moreover the heuristic method of generating instantly decodable
packets is computationally more efficient than the one used in
\cite{ITW2012} and the throughput performance of the proposed method
is approximated using some results from random graph theory.}
\section{System Model and Preliminaries}
We consider a group of $n$ clients $C = \{c_1,\dots,c_n\}$. Each
client $c_i$ is interested in \emph{securely} receiving a distinct
message $x_i\in X=\{x_1,\dots,x_n\}$, where
each message $x_i$ is an element of a finite field  ${\mathbb{F}_q}$.\footnote{In a practical setting, each message can be composed of
multiple elements of ${\mathbb{F}_q}$, where similar operations are
applied to all the elements within a message. The number of elements in each message, $m$,  determines the message and coded packet length.} A transmission scheme is said to be
\emph{secure} if each client $c_i$ is able to obtain its own
message $x_i$, but not the others' messages. We denote the vector of
all messages as $\mathbf{X}=[x_1,\dots,x_n]$.

Our proposed method includes three phases, which are briefly
explained here and will be discussed in detail in the rest of this
paper. For more clear exposition, we will mainly focus on the
general case of multiple unicasts with message secrecy. However, in
Section \ref{sec:Broadcast} we discuss how the results of this paper
include public message broadcast as a special case.
\begin{itemize}
\item \emph{Broadcast}: In the broadcast phase, the messages are
combined using a special form of random linear network coding (RLNC)
\cite{RandomNetCod} and the resulting packets which are denoted by
$P=\{p_1,\dots,p_n\}$, are broadcast to all the clients. The vector
of the generated packets is denoted by $\mathbf{P}=[p_1, \dots,
p_n]$. There are two main differences with practical/traditional
RLNC schemes. The first difference is that the coding coefficients
are not sent along in the header of the packet, as it is often the
case in practice \cite{Chou03practicalnetwork}. The reason for this
is to maintain message secrecy as will be briefly explained shortly
(see Section \ref{sec:KeySharing} for more details). So we can think
of this RLNC as a type of \emph{blind} RLNC.
 The second difference is that coding is done in such a way that not all coded packets in $P$ will be needed to decode a particular message $x_j$ (see also Section \ref{sec:Broadcast}).
 So we can think of this RLNC as \emph{partial} RLNC. Each
client $c_j$ might miss any of the packets in $P$ with a probability
$p_e^j$. For simplicity, we assume $p_e^j=p$ is equal for all the
clients and is constant during the transmission process.
\item \emph{Key Sharing}: The encoding process forms a mapping
between each original message $x_i$ and the set of generated packets
$P$ represented by $x_i=f_i(p_1,\dots,p_n)$. Each function $f_i$ is
privately and securely delivered to the corresponding receiver
(\textbf{e.g. as a private key}), $c_i$, during the key sharing
phase. This function should not be guessable for an eavesdropper
with limited computational resources.
\item \emph{Packet Recovery}: As mentioned earlier, the broadcast
channel between the base station and the clients is modeled as a
packet erasure channel. Consequently, after the broadcast phase,
each client $c_i$ might have received a subset of the packets
represented by $\Gamma_i$. Depending on function $f_i$, client $c_i$
might need to receive some more packets (not necessarily all packets
in $P\backslash \Gamma_i$) to be able to obtain its own message
$x_i$.
\end{itemize}
 In this paper, we consider
two different schemes for packet recovery. The first one is a
centralized scheme where the base station is in charge for
retransmissions. In this case the problem can be formulated as a
standard index coding. In the second scheme, the clients are assumed
to be separated from the base station once each packet is received
by at least one client and the clients cooperate with each other by
exchanging packets to obtain all the packets they need to decode.
This scheme can be regarded as \emph{cooperative or decentralized
index coding}.

The broadcast phase is common between centralized and cooperative schemes, which is the subject of Section
\ref{sec:Broadcast}. In Section \ref{sec:Security}, we propose our key sharing scheme and evaluate its secrecy level. Sections \ref{sec:BaseStation} and
\ref{sec:Cooperation} are dedicated to the packet recovery phase for
the centralized and cooperative schemes, respectively. These two sections include heuristic algorithms for
efficient network coded retransmissions, as well as some theoretical
approximations on the performance of the proposed algorithms and
numerical experiments. The paper is concluded in
Section \ref{sec:Conclusion}.

\section{Broadcast Phase}\label{sec:Broadcast}
In the broadcast phase, messages are combined together using linear
network coding and are broadcast to all the clients. However, the
coefficients are not publicly announced. Instead, each
client is provided with a private set of decoding coefficients which
enable each client to decode and obtain its own message but not the
others'. Each client might need to receive a subset of packets
depending on how the decoding coefficients are defined by the base
station.

We define the decoding matrix $\mathbf{A}$ and denote the $i$-th
row of matrix $\mathbf{A}$ by $\mathbf{A}_i$. In its most general
form, at each row $\mathbf{A}_i$, the number of non-zero elements is
denoted by $r_i$. The rest of elements at each row are set be zero.
The non-zero elements at each row are placed randomly and are chosen
randomly from the finite field $\mathbb{F}_q$. In Section
\ref{sec:KeySharing}, we will describe how the base
station can use a method with low communication overhead for sharing the keys with the
clients.
 %where the
%placement and the value of non-zero elements are matched with the
%private keys of the clients (while the randomness of the elements
%and the placement of the decoding coefficients are still preserved
%as it will be discussed later).
 For simplicity, in this paper we
assume $r_i=r$ is identical for all the clients. The set of indices
of non-zero elements for each client is represented by $I_i$. The
base station solves the following system of linear equations to
generate the set of packets $P$.
\begin{equation}
\mathbf{X}^T=\mathbf{A}\mathbf{P}^T
\end{equation}
The resulting set of packets are broadcast to all the
clients.\footnote{The probability that the decoding matrix
$\mathbf{A}$ is non-singular is given in \cite{ContentSparse}.}
According to the above equation, each message $x_i$ is a linear
combination of a subset of packets $R_i=\{p_j: \forall j\in I_i\}$,
i.e. $x_i=\mathbf{A}_i \mathbf{P}$. In other words, each client
$c_i$ needs to obtain the set $R_i$ to be able to decode and
retrieve its own message $x_i$. Moreover, each client would not be
able to retrieve the message of another client $c_j$ as it does not
have access to $\mathbf{A}_j$. \textbf{Fig. \ref{fig:myfigVeryNew1}
provides a schematic vision of the coding approach proposed in this
paper.}

If each client $c_i$ is aware of $I_i$, it would be able to turn off
its receiver when the other packets are transmitted by the base
station to save its battery. But this will reduce the amount of side
information available at each client. This side information is later
exploited to enable either the base station or the clients to
combine more packets in a single transmission during the packet
recovery phase. Therefore, to achieve a higher bandwidth efficiency,
if the clients are not informed of $I_i$ and keep their radios
switched on during the entire broadcast phase, the total number of
retransmissions by the base station or the clients is expected to be
smaller.

Interestingly, the coding method of the this paper shifts the burden
of matrix inversion operations to the base station, which often has
very large computational resources and only a linear summation over
a finite field is left for decoding at mobile clients. This will
provide more freedom to choose larger field sizes. As it will be
shown later in Section \ref{sec:Security}, operations over larger
field sizes are less vulnerable against brute force guesses by a
wiretapper or an eavesdropping client. Moreover, our proposed scheme
can be easily extended to more general scenarios where each client
is interested in receiving an arbitrary subset of messages. Suppose
a client $c_i$ wishes to receive a subset $U \subseteq X$, where
$U=\{x_{i1},\dots, x_{i|U|}\}$. To apply the proposed coding method
in this paper, the base station only needs to provide all the
corresponding rows of matrix $\mathbf{A}$ required for decoding and
retrieving the set $U$ to client $c_i$, i.e. $\mathbf{A}_i, \forall
i\in \{i_1,\dots,i_{|U|}\}$. This has been extensively discussed in
\cite{ITW2012}

\begin{figure}[t]
\begin{center}
  \includegraphics [scale=0.5]{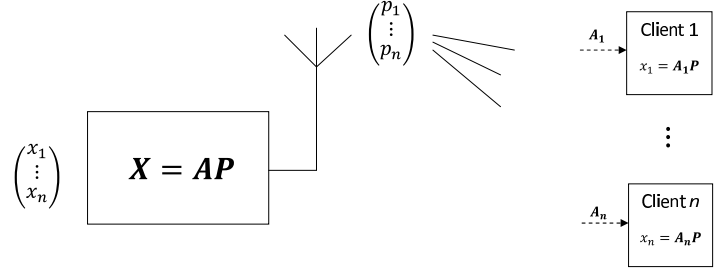}
  \caption{The proposed coding scheme}\vspace{-1em} \label{fig:myfigVeryNew1}
\end{center}
\end{figure}
%[trim = 10mm 120mm 10mm 15mm, clip, scale=.8]

\section{Key Sharing Phase}\label{sec:KeySharing}
As mentioned earlier, the vector $\mathbf{A}_i$ contains crucial
information for decoding at client $c_i$. Therefore, it should be
delivered privately and securely to client $c_i$. We note that for
each $\mathbf{A}_i$, the location of non-zero coefficients and their
values in the finite field are sufficient for its complete
description. The core idea of our key sharing method is as follows.
Each user is securely provided with two unique permutation functions
for deciphering the locations and values of its decoding
coefficients.
 The role of these functions is to map publicly known location and value sets, which are broadcast to everyone by the base station, into privately-known individual keys. In
other words, each client infers a different \emph{meaning}
(decryption) from two common sets of locations and values. Below, we
formalize our proposed method and quantify the computational cost of
breaking the ciphers in a brute-force manner \textbf{if the
probability distribution of the target message over its alphabet is
known to eavesdropper.}

%delivered privately and securely to client $c_i$. The trivial
%solution is to deliver each vector $\mathbf{A}_i$ separately to its
%corresponding client $c_i$ using any secure method
%of transmission. However, to protect message secrecy against brute force attacks in practical systems,  $\mathbf{A}_i$ needs to be changed once in a while.  Each time $\mathbf{A}_i$ changes, it has to communicated privately to $c_i$. Here, we propose a new method to reduce the overhead of separate secure
%transmissions of the decoding coefficients to each user. We first note that the two important pieces of information in each $\mathbf{A}_i$  is the location of non-zero coefficients and their values in the finite field. The key idea is as follows. Each user is securely provided with two permutation functions for its decoding coefficient location and values. The role of these functions is to map publicly known location and value sets into privately-known individual ones. The public location and value sets are broadcast to everyone by the base station. In order to change the decoding coefficients

\begin{mydef}
We denote a permutation of the elements of a set $B$ with $\Pi_B^i:
B\rightarrow B$, which is a one-to-one and covering function that
(randomly) maps an element $\alpha \in B$ to another element in
$\beta \in B$ ($\Pi_B^i(\alpha)=\beta$).
%The reverse function for $\Pi_B^i$ is denoted by $(\Pi_B^i)^{-1}$.
Here $i$ is an arbitrary index, which will be used later to specify a client.
\end{mydef}
\begin{mydef}
We denote the set of all non-zero elements in $\mathbb{F}_q$ by
$Q=\{\alpha_1,\dots,\alpha_{q-1}\}$. Also, $\tilde{N}=\{1,\dots,n\}$
is the set of all positive integers smaller than $n+1$.
\end{mydef}

The decoding coefficients are generated securely as follows:
%Each client $c_i$ is given a pair of permutation functions
%$\Pi_Q^i$ and $\Pi_{\tilde{N}}^i$ as private keys. The matrix of
%decoding coefficients $\mathbf{A}$ used in the broadcast phase, is
%generated according to the set of key pairs $\mathcal{K}=\{(\Pi_Q^1,
%\Pi_{\tilde{N}}^1), \dots, (\Pi_Q^n, \Pi_{\tilde{N}}^n)\}$. In the
%following, the procedure for generating matrix $\mathbf{A}$ is
%described:
\begin{enumerate}
\item Each client $c_i$ is given a pair of unique permutation functions
$\Pi_Q^i$ and $\Pi_{\tilde{N}}^i$ as private keys.
\item The base station generates a subset of $Q$, denoted by $Z_r$, consisting of $r$ elements randomly drawn from $Q$. It also generates a subset of $\tilde{N}$, denoted by $Y_r$,
 consisting of $r$ elements randomly drawn from $\tilde{N}$. The elements of the sets $Y_r$ and
$Z_r$ are represented by $y_i$ and $z_i$ for $i=1,\dots,r$,
respectively. These sets act as public keys and are publicly broadcast to
all the clients.
\item Each row $\mathbf{A}_i$ of the matrix of the decoding coefficients is
generated by setting
$\mathbf{A}_i(\Pi_{\tilde{N}}^i(y_j))=\Pi_Q^i(z_j)$ for
$j=1,\dots,r$.
\item Each client $c_i$ uses the reverse
functions $\Pi_{\tilde{N}}^i$ and $\Pi_Q^i$ to decrypt the set of
decoding coefficients from $Z_r$ and $Y_r$.

%\textbf{why does it need the reverse functions?}
\end{enumerate}

Although the clients are receiving the same sets $Z_r$ and $Y_r$
from the base station, the decryption process in each client $c_i$
results in different set of decoding coefficients since the
functions $\Pi_{\tilde{N}}^i$ and $\Pi_Q^i$ are different for each
client. In other words, we use a private key ciphering method to
deliver the decoding coefficients to each client, where the vector
of decoding coefficients, $\mathbf{A}_i$, can be envisaged as a
secret key for the client $c_i$ \textbf{to decode its own message}.
\textbf{The main advantage of using this method is to reduce the
overhead of updating the decoding coefficients as each time only
broadcasting one pair of common sets $Z_r$ and $Y_r$ is sufficient
for updating the decoding coefficients of all the clients. Updating
the decoding coefficients might be necessary to increase the level
of secrecy against brute force guess as it will be discussed in the
following.}

%Since the proposed method for private
%transmission of decoding coefficients is considerably low overhead,
%it would be easier to change the matrix $\mathbf{A}$ once a while
%(say after multiple rounds of transmission) to increase the secrecy
%of the system against brute-force attacks.

\subsection{Computational Cost of Brute-force Guesses}\label{sec:Security}
One main concern of this paper is to provide a secure method of
transmission for multiple unicasts over a shared broadcast channel
while the advantages of cooperation and/or network coding is
incorporated. That this type of privacy can be provided by hiding
the decoding coefficients and transmitting the corresponding
decoding coefficients privately to each client using any secure
method. However, one might be concerned about the privacy level that
can be guaranteed. In other words, how easy it is for one of the
clients, say $c_i$ (or an external eavesdropper who is tapping the
channel) to obtain some information about an individual message
$x_j$.

Clearly there is not an algebraic method to find the solution of a
system of linear equations while the coefficients are not known.
However, since the operations are performed over a finite field, the
total number of possibilities for choosing the decoding coefficients
is limited by the size of the field $q$, the number of clients $n$
and the \emph{coding rate}, defined as the number non-zero elements
$r$ in each $\mathbf{A}_i$. In other words the adversary may run a
set of brute-force trials to examine all the possibilities of
decoding coefficients to observe a message which would be most
likely to be an original message. The following theorem determines
the average computational resources needed to reveal a message by
applying a finite set of trials in a brute force scenario.
\begin{theorem} \label{computational}
An eavesdropping client $c_i$ needs to try
$\frac{\min\{(q-1)^r,(q-1)!n!\}+1}{2}$ on average for $q>2$, to make
a correct guess about the decoding coefficients $\mathbf{A}_j$ of
the client $c_j$, $j\neq i$.
\end{theorem}
\begin{proof}
The problem is equivalent
to the following problem. Suppose that there are $M$ keys out of $N$ keys that can open a locked
door. It can been shown (using some combinatorics) that
on average $\frac{N+1}{M+1}$ trials are required to find the first correct
key without replacement of the keys. In our problem, there is only one correct key (decoding coefficient) for each client. Therefore, $M = 1$. On order to find $N$, we proceed as follows.

As will be seen in the next two sections, in order for the clients
to obtain their missing packets for decoding, they need to reveal
the location of non-zero decoding coefficients either to the base
station or to other clients using public broadcast channels. This
will help to identify the \emph{Has}, \emph{Lacks} and \emph{Wants}
set of each client for devising an efficient network coded
retransmission phase \textbf{(Definition of the mentioned sets will
be provided in Section \ref{sec:BaseStation})}. Each non-zero entry
in $\mathbf{A}_i$ can take $q-1$ different values from the finite
field $\mathbb{F}_q$. Therefore, $c_i$ needs to make guesses over a
space of $(q-1)^r$ different possibilities. However, if $(q-1)!n!<
(q-1)^r$ , then the attacker would prefer to make guesses over the
permutation functions $\Pi_{\tilde{N}}^i$ and $\Pi_Q^i$ with $n!$
and $(q-1)!$ possibilities each. Therefore, the total number of
possible keys is $\min\{(q-1)^r,(q-1)!n!\}$ and the theorem is
proved.
\end{proof}

\textbf{Moreover, in an information theoretic sense}, two types of
security can be considered. To be unconditionally secure
\cite{ShannonSecrecy, SecureNetCod, SecrecyLaszlo}, the length of
the key should be equal to the message. In that case, since the
decoding coefficients act as key, unconditional security implies
that after each round of transmitting single- element messages (a
message with one element from the finite field $\mathbb{F}_q$), a
new set of decoding coefficients should be generated and broadcast
to all the clients. If the unconditional security is relaxed to a
weaker definition of security similar to \cite{WeaklySecure}, then
it is possible to transmit longer messages (more than one element of
$\mathbb{F}_q$) as the eavesdropper cannot obtain any meaningful
information about the individual messages, although some information
about the joint distribution of the messages might be leaked.
\textbf{This notion has been partially discussed in \cite{ITW2012},
where it is proven that our suggested scheme is weakly secure in
general and is unconditionally secure if the decoding coefficients
are refreshed after each round of transmission both for a finite
field of size $q=2$, i.e. $\mathbb{F}_2$.}

%\textbf{Another issue which is worthwhile to be mentioned here, is
%that whoever have the decoding coefficients vector $\mathbf{A}_i$
%would be able to obtain message $x_i$. Therefore, the proposed
%scheme is not limited to the case of multiple unicasts and can be
%extended to more general scenarios. Any client $c_i$ might be
%privileged to have access to a subset $\tilde{X}_i\subseteq X$ where
%the corresponding permutation function pairs (for those messages in
%$\tilde{X}$ should be provided to $c_i$. As a special case,
%$\tilde{X}_i=X, \forall i\in \{1,\dots,n\}$ which is a broadcast
%scenario. In this case all the clients would be able to decode all
%the messages but the messages are protected against external
%eavesdroppers}

\section{Centralized Packet Recovery via the Base Station} \label{sec:BaseStation}
 The downlink channel between the base station and
each client is modeled as a packet erasure link. After the initial
broadcast phase each client $c_i$ might have missed each packet with
probability $p$. From the set of missing packets by client $c_i$
those packets which are in the set $R_i$ should be retransmitted
either by the base station or a neighboring client. In this section
we analyze retransmissions by the base station. We use a matrix $S$
to denote the reception status of each packet for all the clients.
Each entry $s_{ij}$ in matrix $S$ denotes the reception status of
packet $p_j$ for client $c_i$ which can take one of the following
three values:
\begin{itemize}
\item $s_{ij}=1$ if client $c_i$ initially holds packet $p_j$. The
set of packets initially received by $c_i$ is represented by
$\Gamma_i\subseteq P$ (\emph{Has} set).

\item $s_{ij}=2$ if client $c_i$ initially does not hold packet
$p_j$ but needs that packet to recover its own message. The set of
such elements for client $c_i$ is denoted by
$\Omega_i=\bar{\Gamma}_i \cap R_i$ (\emph{Wants} set).
\item $s_{ij}=3$ if client $c_i$ neither holds the packet $p_j$ nor
needs it to recover its own message. These elements are shown by
$\Psi_i=\bar{\Gamma}_i \cap \bar{R}_i$.
\end{itemize}
\textbf{The \emph{Lacks} set is the set of packets which have not
been received by client $c_i$, i.e. the set
$\bar{\Gamma_i}=\Omega_i\cup \Psi_i$. It should be noted that the
mentioned sets are updated during the packet recovery phase, however
for simplicity, we do not use an index of time for these sets.}

We assume each client sends an ACK or NACK feedback to the base
station in an error-free phase to acknowledge reception of each
packet. Therefore the base station can form the matrix $S$.

 This problem is an instance of index coding problem where each
client $c_i$ holds a subset of packets denoted by $\Gamma_i\in P$
and is interested in receiving another subset $\Omega_i \in P$ where
$\Gamma_i\cap \Omega_i=\emptyset$, but $\Gamma_i\cup \Omega_i$ is
not necessarily equal to $P$. Finding an optimal transmission
solution for index coding in its general form has been proven to be
NP-hard \cite{IndexCodMain,IndexCodingNP}. In
\cite{IndexCodingHeuristics} a heuristic algorithm is proposed to
find a possibly sub-optimal solution where at each transmission an
instantly decodable packet is generated. An instantly decodable
packet is a linear combination of a subset of packets which is
immediately decodable by some of the clients and is discarded by the
others. In other words, using our notations in this paper, if client
$c_i$ initially holds a subset $\Gamma_i\subset P$ and receives a
linear combination of the packets in $\Gamma_i\cup \{p_{\ell}\in
\bar{\Gamma}_i\}$ for some $1\leq \ell \leq n$, it would be able to
instantly decode and detach $p_{\ell}$. Moreover, \cite{SamehDelay}
suggests an improved version of the algorithm in
\cite{IndexCodingHeuristics} to minimize the broadcast completion
delay, where the clients with larger number of packet demands are
given priority (using a weighting mechanism) while the instantly
decodable packets are generated. We apply \textbf{the same}
heuristic as in \cite{SamehDelay} to minimize the total number of
transmissions. \textbf{However, we use a different analytical
approach to approximate the number of required transmissions using
the mentioned heuristic.}
\subsection{Algorithm}
At first, the index coding problem is converted to an equivalent
graph $G(V,E)$. Each packet $p_j\in \Omega_i$ is mapped to a
distinct vertex denoted by $v_{ij}$, in other words, there exists a
vertex for each packet in the \emph{Wants} sets of each client. Two
vertices $v_{ij}$ and $v_{kl}$ are connected by an edge if one of
the following conditions holds:
\begin{itemize}
\item Vertices $v_{ij}$ and $v_{kj'}$ are connected if $j=j'$. In
other words if clients $c_i$ and $c_k$ are seeking the same packet
$p_j$.
\item Vertices $v_{ij}$ and $v_{kl}$ are connected if $p_j\in
\Gamma_k$ and $p_l \in \Gamma_i$. In other words, each client holds
the demanded packet by the other one.
\end{itemize}
A subset $\hat{V}=\{v_{i_1j_1},\dots,v_{i_mj_m}\} \subseteq V$ for
some $m\leq n$ forms a clique if any two vertices in $\hat{V}$ are
connected by an edge in $G$. Such a clique is equivalent to a subset
of packets denoted by $\mathcal{P}_V \subseteq P$. XORed version of
the packets in this subset, i.e. $\bigoplus_{p_u\in \mathcal{P}_V}
p_u$ is instantly decodable by the target clients.

We also denote the neighborhood of a vertex $v_{ij}$ as the set of
vertices which are connected to $v_{ij}$ by an edge and we denote it
by $\mathcal{N}_{ij}=\{v_{ts}:(v_{ij},v_{ts})\in E\}$. We define the
weight of each vertex $v_{ij}$ as follows:
\begin{equation} \label{eq:weight}
w_{ij}=|\Omega_i|\sum_{v_{ts}\in \mathcal{N}_{ij}}|\Omega_t|
\end{equation}
The algorithm for generating instantly decodable packets to transmit
to all the clients by the base station is given in Algorithm.
\ref{algorithm1}.
\begin{algorithm}[h]
\caption{Packet Recovery via Base Station}
 \label{algorithm1}
 \textbf{while} $V\neq \emptyset$ \tab \%\emph{The set of vertices}\\
   $V_{temp}=V$\\
   $\hat{V}=\emptyset$\\
   Calculate $w_{ij}$ 's for $G$. \tab \%\emph{According to equation \eqref{eq:weight}}\\
   \tab \textbf{while} $V_{temp}\neq \emptyset$ \tab  \% \emph{Nodes with higher weights added to the clique}\\
    \tab  $v^*=\mathrm{argmax}_{w_{ij}}\{V_{temp}\}$\\
    \tab $V_{temp}\leftarrow V_{temp}\setminus v^*$\\
    \tab \tab  \textbf{If} $\hat{V}\cup v^*$ is a clique\\
      \tab \tab   $\hat{V}\leftarrow \hat{V}\cup v^*$\\
    \tab \tab     \textbf{End if}\\
      \tab   \textbf{End while}\\
         $P_{\hat{V}}=\{p_j: v_{ij}\in \hat{V}\}$\\
         Broadcast the packet $p_B=\bigoplus_{p_g\in P_{\hat{V}}} p_g$ to all the clients.\\
         \textbf{For} $i=1:n$ \tab \% \emph{Matrix $S$ is updated for some clients} \\
          \tab  \textbf{If} client $c_i$ received the packet $p_B$\\
            \tab \tab    \textbf{If} $v_{ij}\in \hat{V}$\\
             \tab   $s_{ij}=1$, $\Gamma_i \leftarrow\Gamma_i\cup \{p_j\}$, $\Omega_i \leftarrow\Omega_i\setminus \{p_j\}$, $V\leftarrow V\setminus
             \{v_{ij}\}$\\
             \tab \tab   \textbf{End if}\\
             \tab   \textbf{End if}\\
                \textbf{End for}\\
\textbf{End while}
\end{algorithm}
In Alg. \ref{algorithm1}, at each round of the algorithm the
vertices are ordered according to their weights. The vertex with
maximum weight is considered as the first one in a list of vertices.
The other vertices are added to the list if they form a clique with
the \textbf{exiting} ones in the list. Therefore, at each round of
the algorithm, a list of vertices is generated which is equivalent
to an instantaneously decodable packet for the clients in the clique
and is broadcast to all clients. Each receiver might receive the
retransmitted packet with probability $1-p$ (We assume the
probability of erasure remains constant during broadcast and
recovery phases). The sets $\Gamma_i$, $\Omega_i$ and matrix $S$ are
updated for those target \footnote{Here, by the target user we mean
those users which are able to obtain one packet in their
\emph{Wants} set $\Omega_i$. Later in this section, we will show
that if any user which is able to decode and obtain a packet within
its \emph{Lacks} set is targeted the throughput would be higher,
however the energy consumption of the users increases as they have
to keep their radios switched on for more transmissions} clients
that received the retransmitted packet packet without error. Those
vertices for which their equivalent missing packet is received by
the corresponding client are removed from $G(V,E)$.
\subsection{Analysis of throughput}
The algorithm discussed in this section tries to find the clique
with maximum size at each round of transmission (Although this can
not be guaranteed and the algorithm may find suboptimal size
cliques). Graph $G(V,E)$ in our problem, \textbf{can be modeled as a
random graph by nature}.
 %(as the graph is resulted from the decoding matrix
%which identifies the required packets and is a random matrix as
%defined before and the erasure of packets by the channel which is a
%random process itself).
 A random graph $\mathcal{G}(N,\pi)$ is used
to model and represent the actual graph $G(V,E)$ and is defined with
two parameters: the number of vertices (denoted by $N$) and the
probability that there exists an edge between two specific vertices
(denoted by $\pi$). \textbf{Parameters $N$ and $\pi$ in our problem
will be identified soon in this section.}

Each client is interested in receiving a specific set of $r$
packets. Since the probability of erasure is assumed to be constant
for all the transmissions, the number of missing packets by each
client is a Bernoulli random variable. Therefore, the average number
of demanded packets by each client is $rp$. Consequently, the
expected number of vertices in $\mathcal{G}$ is $N=nrp$. We use this
average instead of the actual number of vertices which is a random
variable (summation of $n$ binomial random variables) which
simplifies our analysis but at the price of losing some precision.

If the retransmissions were error-free, the problem could be modeled
as a clique partitioning problem \cite{IndexCodingHeuristics}.
However, since the reception of retransmitted packets at all the
target receivers cannot be guaranteed, one cannot assure that the
entire set of vertices in the corresponding cliques are removed. As
a consequence, the problem would be different from the standard
clique partitioning problem. The following theorem provides an
approximation for the total number of required transmissions
incorporated in the recovery phase.

%In 1976, Bollobas and Erdos proved that the size of a maximum size
%clique in $\mathcal{G}$ tends to $\frac{2\log
%N}{\log(\frac{1}{\pi})}$ as $N\rightarrow \infty$. Therefore, for
%smaller values of $N$ the theorem does not provide the precise value
%of the largest clique in the random graph. However, it can be used
%as an approximated assessment to find the size of the largest
%clique. As it can be observed in the algorithm, finding the largest
%clique does not guarantee the removal of all its vertices (or
%equivalently satisfaction of all the intended receivers), as
%erasures might prevent the generated packet to be delivered to one
%or some of the intended destinations. If the size of the

In the following, we provide an approximate assessment for the total
number of required transmissions for large graphs (as $N\rightarrow
\infty$) .
\begin{theorem}
\label{thm1}
 The total number of required transmissions $T$ is
approximated by:
\begin{equation*}
\begin{split}
T&=\max{t}\\
\mathrm{s.t. }\\
N_{t+1}&=N_t-2(1-p)\frac{\log N_t}{\log\frac{1}{\pi}}\\
\hat{p}(t+1)&=\hat{p}(t)-\frac{N_{t+1}-N_t}{nr}\\
N_t&>0\\
N_0=N,
\pi&=\frac{n-1}{n}[(1-\hat{p}(t))^2+\hat{p}(t)^2],\hat{p}(0)=p
\end{split}
\end{equation*}
\end{theorem}
\textbf{where $N_t$ is the size of the remaining number of vertices
found by the algorithm at the end of $t$'th transmission and
$\hat{p}(t)$ approximates the probability of missing a packet (or
equivalently the probability of a packet to be in the lacks set of a
client) at round $t$, where $\hat{p}(0)$ is set to be $p$ before the
algorithm starts.}
\begin{proof}
At each round, the algorithm searches to find the largest clique for
removal (However maximality cannot be guaranteed as the algorithm is
a heuristic). In 1976, Bollabas and Erdos proved
\cite{RandomGraphClq} that the size of the largest clique in a large
random graph $\mathcal{G}(N,\pi)$ is as follows (as $N\rightarrow
\infty$):
\begin{equation}
\mathcal{X}_N=\frac{2\log N}{\log\frac{1}{\pi}}
\end{equation}
At each transmission, a combination of $p_i$'s is generated based on
the clique found by the algorithms and is broadcast to all the
clients. However each client receives a packet with probability
$1-p$.Therefore, at $t$'th transmission, $(1-p)\mathcal{X}_{N_t}$
nodes are removed. The algorithm continues until all the nodes are
removed from $\mathcal{G}$. This results in the specified recursive
relation between $N_t$ and $N_{t+1}$ after $t$'th transmission.
Therefore $T=\max{t}$ indicates the number of transmissions
incorporated.

 The probability $\pi$ that two nodes are connected
together is obtained as follows:

\begin{equation}
\begin{split}
\pi=P((v_{ij},v_{kl})\in E)=P[(p_j\in \Gamma_k) , (p_l\in \Gamma_i),
(i\neq k)]
 + P[((p_j=p_l)\notin(\Gamma_i\cup \Gamma_k),
(i\neq k)]\\
=P(i\neq k)[P((p_j\in \Gamma_k)P(p_l\in \Gamma_i)
+P(p_j\notin\Gamma_i)P(p_j\notin\Gamma_k)]=
\frac{n-1}{n}[(1-p)^2+p^2]
\end{split}
\end{equation}
In other words an edge exists between two nodes in $\mathcal{G}$ if
either they demand the same packet or if each of them holds the
packet demanded by the other one.

\textbf{At each round, the number of packets which are made known to
the clients is $N_{t+1}-N_t$ out of the total number of $nr$ packet
receptions required for all clients to decode. We can claim that
this number is virtually proportional to $\frac{n}{r}(N_{t+1}-N_t)$
out of $n^2$ packets considering the total number of possible packet
receptions at all clients ($n$ packets for $n$ clients). Therefore,
to update the probability of missing a packet, one might roughly
expect that
$\frac{\frac{n}{r}(N_{t+1}-N_t)}{n^2}=\frac{N_{t+1}-N_t}{nr}$ should
be subtracted from $p(t)$ at round $t+1$ resulting in
$\hat{p}(t+1)$. $\pi$ is updated at accordingly at each round $t$. }
\end{proof}
In the next subsection, the effect of parameters $r$ and $p$ on the
\textbf{total number of transmissions} of the system is examined via
simulations. Moreover, theorem \ref{thm1} is verified via
simulations.
\subsection{Numerical Experiments}
In our numerical experiments we follow two main goals:
\begin{itemize}
\item The effect of the erasure probability $p$ and the coding rate
$r$ on the performance of the system (total number of required
transmissions) is evaluated. It is expected that the observed
quality of service by each client would be proportional to the
performance of the entire system as the sets $R_i$ are identified
randomly (\textbf{$R_i$ is the set of packets client $c_i$ needs to
receive to be able to obtain its own message}). Larger values of $r$
leads to more required transmissions as each client might need to
receiver more packets to decode and consequently more packets are
likely to be missed by each client. However, there is a trade-off
between the coding rate $r$ and secrecy against computational
(brute-force attacks). Therefore, it is important to find a
reasonable range for $r$ considering the computational capacity of
an eavesdropping client (See section \ref{sec:Security}).
\item Theorem. \ref{thm1} is verified by numerical experiments. As
mentioned in the proof of theorem. \ref{thm1}, the results are more
precise for larger number of nodes in the random graph
$\mathcal{G}$. Therefore, an approximation error is observed between
simulation results and theory.
\end{itemize}
We have run $1000$ rounds of simulation for a set of $n=20$ clients.
We have considered five different values of erasure probability:
$p=0.1$, $p=0.2$, $p=0.3$, $p=0.4$ and $p=0.5$. Operations are
assumed to be performed in a sufficiently large field size. We
further assume that the probability of erasure is constant during
the initial broadcast phase and packet recovery and is identical for
all the clients and the feedbacks from each client is correctly
received by the base station. Fig. \ref{fig:myfig1} shows the
average ratio of the total number of successfully decoded packets to
total number of transmitted packets. Since we have assumed that each
client is interested in one distinct message, the number of packets
that should be successfully decoded at the end is $n$. Therefore, if
we denote the number of packets transmitted during the packet
recovery phase by $T$, Fig. \ref{fig:myfig1} shows $\frac{n}{n+T}$
for different values of $r$ for the corresponding values of erasure
probability $p$.
\begin{figure}[t]
\begin{center}
  \includegraphics [scale=0.5]{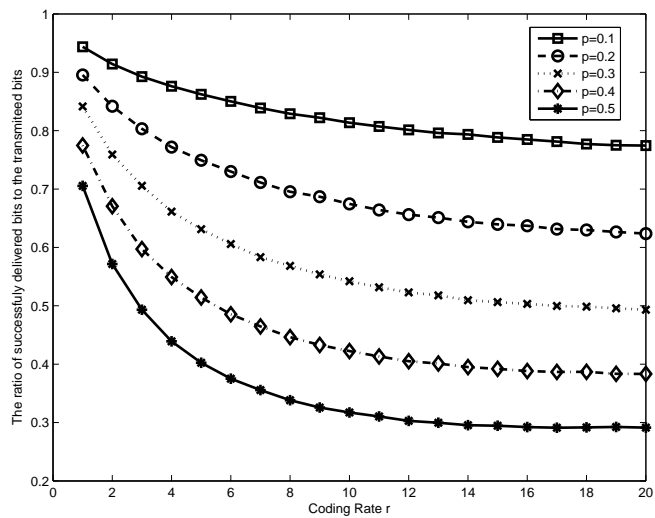}
  \caption{Total number of transmissions for different values of $r$.}\vspace{-1em} \label{fig:myfig1}
\end{center}
\end{figure}
%[trim = 10mm 65mm 10mm 85mm, clip, scale=.47]
Fig. \ref{fig:myfig2} compares the result of simulations with the
result of theorem \ref{thm1}. We have run $1000$ rounds of
experiments for $n=100$ clients and $r=60$. In each experiment, the
number of removed nodes have been recorded and the average value of
these numbers are shown in the figure. As discussed earlier, because
of the erasure over the downlink channel between the base station
and the clients, all the nodes within a clique found by the
algorithm are not guaranteed to be removed. Also, as mentioned
earlier, there is an understandable gap between simulations and
theory as the theorem applies to large number of nodes
($N\rightarrow\infty$).
\begin{figure}[t]
\begin{center}
  \includegraphics [scale=0.5]{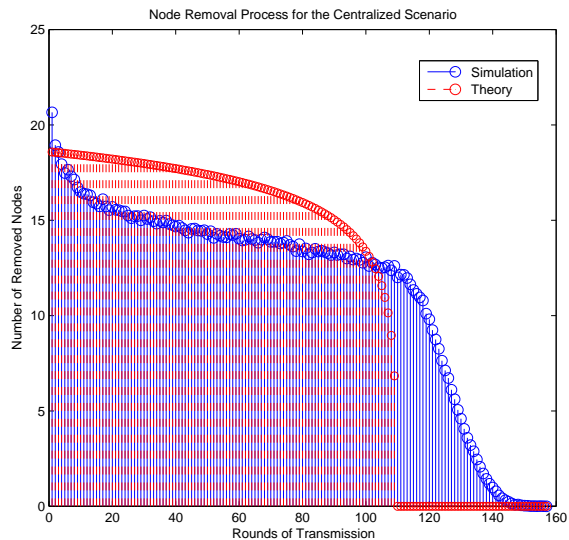}
  \caption{Theory vs. simulations for the packet recovery via base station.}\vspace{-1em} \label{fig:myfig2}
\end{center}
\end{figure}
%[trim = 10mm 65mm 10mm 85mm, clip, scale=.47]
 Another interesting
observation also has been made to examine a modified version of
algorithm \ref{algorithm1}. In the basic version, the algorithm only
updates those receiver who have been \emph{target recipients} of the
transmitted packet at each round, i.e. those $c_i$'s which are able
to decode a packet in their $\Omega_i$'s. However, in the modified
version, the algorithm updates all the clients which are able to
instantly decode and obtain a packet regardless of if that belongs
to their demanded packets or not (the decoded packet might belong to
$\Psi_i$ or $\Omega_i$). This policy helps each client $c_i$ to
extend its \emph{has} set faster than the current version as more
packets are decoded and buffered by each client. This will obviously
increase the chance of coding more packets together at each round of
transmission (to produce an instantly decodable packet). However, it
should be noted that the average energy consumed by the receiver
device at each client is different for the basic and modified
versions. We assume that each client is informed of which packets
are coded together priori at the beginning of each transmission
round. Therefore, each client can turn off its radio if the packet
is not intended for that client. Hence, in the modified version each
client should listen to those packets  (and consume energy) which it
can decode, but in the current version each client $c_i$ only
listens to those packets which it can decode and are in its set
$\Omega_i$. \textbf{Throughput performance of the modified version
is difficult to be analyzed theoretically, however we have evaluated
its throughput and energy consumption via simulations}. We have run
$1000$ rounds of experiments for both versions where $p=0.3$ and $r$
varies from $r=1$ to $r=20$ for $n=20$ clients. Fig.
\ref{fig:myfigNew1} shows the average total number of transmissions
and Fig. \ref{fig:myfigNew2} compares the two versions for the
average number of packet receptions by each client, i.e. the number
of transmission rounds that a client should turn on its radio to
receive a packet. As it can be inferred from Fig.
\ref{fig:myfigNew1} and Fig. \ref{fig:myfigNew2}, the modified
version performs better in terms of total number of transmissions,
however this would be at the price of more energy consumption by the
receivers.
\begin{figure}[h]
\begin{center}
  \includegraphics [scale=0.5]{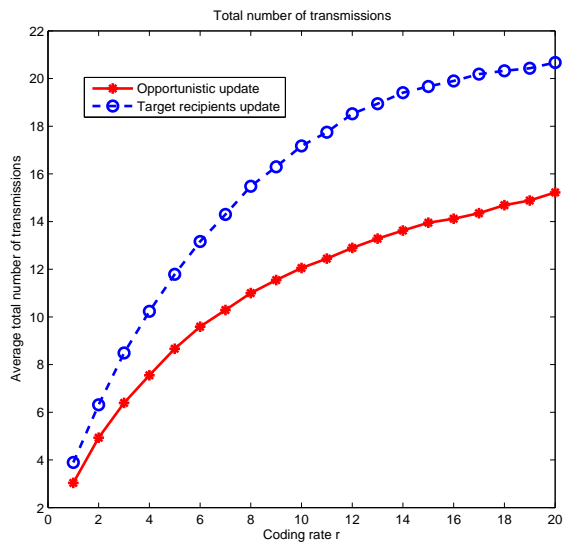}
  \caption{Total number of transmissions for the current and modified version of the algorithm}\vspace{-1em} \label{fig:myfigNew1}
\end{center}
\end{figure}
\begin{figure}[h]
\begin{center}
  \includegraphics [scale=0.5]{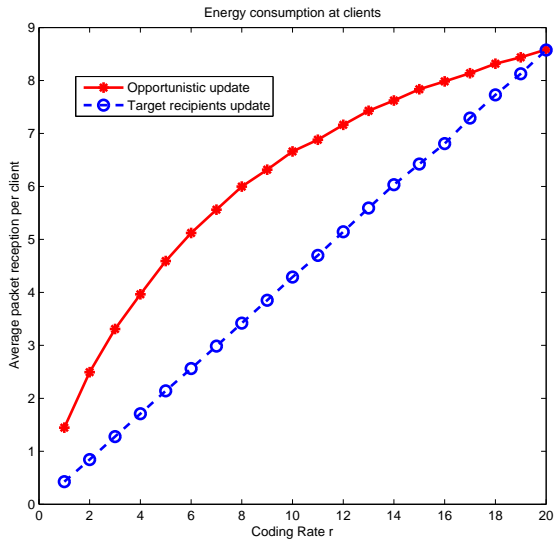}
  \caption{Average number of packet receptions by a client for the current and modified version of the algorithm}\vspace{-1em} \label{fig:myfigNew2}
\end{center}
\end{figure}
%[trim = 10mm 65mm 10mm 85mm, clip, scale=.54] [trim = 10mm 65mm 10mm 85mm, clip, scale=.54]

\section{Decentralized Packet Recovery via Cooperation} \label{sec:Cooperation}
 In Section \ref{sec:Broadcast}, we assumed
 that the base station is in charge for packet recovery
(retransmissions). In this section, we consider a different
possibility: clients can be separated from the base station and work
together to obtain their required packets. Short range links among
the clients are fast, cheap and reliable. More importantly, a
considerable amount of base station bandwidth which was supposed to
be used for retransmissions is released and can be used for other
purposes.

We assume the clients can be divided into clusters such that, each
client is able to send packets to the other clients within the same
cluster and each client receives a packet transmitted by a neighbor
with a probability $p'$ (which is different from the erasure
probability $p$ over the downlink). We denote the $m$-th cluster in
the network with $C^m\subseteq C$. We assume each client can belong
to only one cluster.
%This setting is similar to the assumption made
%by \cite{ITW2010, ISIT2010, NetCod2011}

The broadcast and key sharing phases are identical to centralized
scenario. Also the matrix $S$ has a similar definition. However, in
the cooperative scenario each client should inform the other clients
in the same cluster about the identity of packets it has received
and packets it is looking for. If we denote the $i$-th row of matrix
$S$ with $S_i$, then each client $c_i\in C^m$, should send $S_i$ to
its neighbors in $C^m$. We assume these transmissions are reliable,
so each client is aware of the what each neighbor has received and
which packets it needs to be able decode its message.

\subsection{Algorithm}
The problem discussed in this section can be called
\emph{cooperative index coding problem}. The problem is formally
defined as follows:
\begin{mydef}[Cooperative Index Coding Problem]
We consider the group of clients (called as a cluster) $C^m\subseteq
C$. Each client $c_i\in C^m$ has received a subset $\Gamma_i\in P$
and requires a subset $\Omega_i\in P$. We further assume
$\bigcup_{\forall i: c_i\in C^m} \Gamma_i=P$, i.e. the clients in
$C^m$ hold the entire set of packets collectively. The clients
cooperate with each other by exchanging functions of their packets
to help each $c_i$ to obtain $\Omega_i$. The minimum number of
transmissions is desired.
\end{mydef}
 %Since each client has received a subset of packets, a client
%might not be able to generate an instantaneously decodable packet
%equivalent to a maximal clique in graph $G(V,E)$. Based on this
%observation, a heuristic solution is to generate a local graph in
%addition to the main graph $G(V,E)$ the for each client as follows
%(The main graph $G(V,E)$ is generated as described in section
%\ref{sec:Broadcast}):

Since each client has only received a subset of packets, it would be
only able to generate combinations of packets within its own set
$\Gamma_i$. \textbf{This leads us to the idea of generating local
graphs for all the clients, where each local graph is generated for
the same set of vertices. However, only those edges which their both
end vertices are mapped from packets within the has set of a client
are included in its local graph. Therefore, removing a clique from a
local graph corresponding to a specific client $c_u$ is translated
to generating a coded packet by client $c_u$ which is instantly
decodable by others. In the following the process of generating
local graphs is formally described. The main graph $G(V,E)$ is
generated as described in Section \ref{sec:BaseStation}.}

%This leads us to the idea of generating local graphs in addition to
%the main graph $G(V,E)$ where each local graph corresponding to
%client $c_i$ is formed by considering only packets within
%$\Gamma_i$. Therefore, each clique for such a local graph is
%equivalent with a packet that $c_i$ is able to generate and transmit
%to the others.

For each client $c_u$ a corresponding local graph is represented by
$G_u(V,E_u)$. Each packet $p_j\in \Omega_i$ is mapped to a vertex
$v_{ij}$. The set of all vertices is denoted by $V$ and is identical
for all the local graphs. Two vertices are connected by an edge in
$G_u$ if one of the following conditions hold:
\begin{itemize}
\item There exist an edge $e_{ijkj'}^u \in E_u$ between vertices $v_{ij}$ and
$v_{kj'}$ if $j'=j$ and $i\neq k$ in the main graph $G(V,E)$ and
$p_j\in \Gamma_u$.
\item There exist an edge $e_{ijkl}^u \in E_u$ between vertices $v_{ij}$
and $v_{kl}$ if $p_j\in \Gamma_k$ and $p_l\in \Gamma_i$ for $i\neq
k$ and also $p_j, p_l\in \Gamma_u$.
\end{itemize}
In other words, $E_u\subseteq E, \forall u\in \{1,\dots,n\}$. A
subset $\hat{V}^u = \{v_{i1j1},\dots,v_{i\ell j\ell}\}$ for some
$\ell\leq n$ forms a clique in a local graph $G_u$ if any two
vertices $\hat{V}^u$ are connected by an edge in $G_u$. In that case
client $c_u$ is able to generate an instantaneously decodable packet
for the clients $c_{i1}, \dots,c_{i\ell}$.

Each local graph is weighted using the same criteria in Section
\ref{sec:Broadcast}. The neighborhood of a vertex $v_{ij}$ in graph
$G_u$ is denoted by $N_{ij}^u=\{v_{ts}:(v_{ij}, v_{ts})\in E_u\}$.
We define the weight of each vertex $v_{ij}$ as follows:
\begin{equation}\label{eq:weightccop}
w_{ij}^u=|\Omega_i|\sum_{v_{ts}\in \mathcal{N}_{ij}^u}|\Omega_t|
\end{equation}
Alg. \ref{Alg2} provides a heuristic to generate an instantly
decodable packet by one of the clients at each round of
transmission. The algorithm attempts to remove as many nodes as
possible at each round. Alg. \ref{Alg2} is a modified version of
Alg. \ref{algorithm1} suitable for the cooperative scenario. At each
round, the algorithm selects the largest clique among the maximal
cliques found by the clients (maximum of the maximums) for
transmission where each client applies one round of Alg.
\ref{algorithm1} to its local graph. For simplicity, we have used
indices $i$,$j$ and matrix $S$ locally for an arbitrary cluster
$C^m\subseteq C$.
\begin{algorithm}[h]
\caption{Packet Recovery via Cooperation} \label{Alg2}
 \textbf{while} $V\neq \emptyset$\\
    $M=[]$\\
    $V_{temp}=V$\\
   $\hat{V}^u=\emptyset$, $\forall u=1,\dots,|C^m|$\\
   Calculate $w_{ij}$ 's for $G$ and $G_u$ for $u=1,\dots,|C^m|$ \tab \%\emph{According to Eq. \eqref{eq:weight} and Eq. \eqref{eq:weightccop}, respectively.}\\
   \textbf{For} $u=1:|C^m|$\\
   \tab \textbf{while} $V_{temp}\neq \emptyset$\\
    \tab  $v^*=\mathrm{argmax}_{w_{ij}}\{V_{temp}\}$\\
    \tab $V_{temp}\leftarrow V_{temp}\setminus v^*$\\
    \tab \tab  \textbf{If} $\hat{V}\cup v^*$ is a clique\\
      \tab \tab   $\hat{V}\leftarrow \hat{V}\cup v^*$\\
    \tab \tab     \textbf{End if}\\
      \tab   \textbf{End while}\\
      $M(u)=|(\hat{V}^u)|$\\
      \textbf{End for}\\
      $u_0=argmax_{u} M$\\
      $\hat{V}=\hat{V}^{u_0}$
         $P_{\hat{V}}=\{p_j: v_{ij}\in \hat{V}\}$\\
         Broadcast the packet $p_C=\bigoplus_{p_g\in P_{\hat{V}}} p_g$ to all the clients.\\
         \textbf{For} $i=1:|C^m|$\\
          \tab  \textbf{If} client $c_i$ received the packet $p_C$\\
            \tab \tab    \textbf{If} $v_{ij}\in \hat{V}$\\
             \tab   $s_{ij}=1$, $\Gamma_i \leftarrow\Gamma_i\cup \{p_j\}$,$\Omega_i \leftarrow\Omega_i\setminus \{p_j\}$, $V\leftarrow V\setminus
             \{v_{ij}\}$\\
             \tab \tab   \textbf{End if}\\
             \tab   \textbf{End if}\\
                \textbf{End for}\\
\textbf{End while}
\end{algorithm}

\subsection{Analysis of throughput}
The main issues which differentiate the cooperative (decentralized)
scenario from the centralized scenario can be explained as follows:
\begin{itemize}
\item There is not a central entity which holds the entire set of
packets. Each client has received a subset. If $\Gamma_i \cup
\Omega_i=P$ (or $r=n$) for all $i=1,\dots,n$, the problem reduces to
the cooperative data exchange problem discussed in \cite{ITW2010,
ISIT2010, NetCod2011} which has a linear programming solution if the
packets are allowed to split into smaller fractions. However, in its
general form, as discussed earlier, it can be explained as
cooperative index coding problem.
\item Since we assumed that for each packet $p_j$ there exist at least one
client $c_i$ that $p_j\in \Gamma_i$, if in the broadcast phase a
packet is not received by any of the clients, it should be
retransmitted once again by the base station which affects the
initial distribution of packets in a statistical sense (in the
second or later transmissions of packet $p_j$ there is a chance that
more clients receive it).
\item The probability $p'$ that a packet which is transmitted by a client
$c_i$, is received by $c_k$, is different from the probability of
erasure over the downlink between the base station and each client.
\end{itemize}
As the first step, the probability that the client $c_i$ would have
received a specific packet $p_j$ should be clarified (which is
denoted by $p_{eff}$. The base station stops transmitting a specific
packet $p_j$ if it would have been received by at least one of the
clients. We assume the base station receives perfect feedback from
the clients on the reception status of each packet. The probability
that none of the clients have not received packet $p_j$ after $t$
transmissions by the base station, and client would have received
packet $p_j$ at $(t+1)$'th transmission in a cluster of size $|C^m|$
is $(1-p)(p^{|C^m|})^t$. Theoretically $t$ can be infinite,
therefore:
\begin{equation}\label{myequa}
\begin{split}
1-p_{eff}&=\\
(1-p)+(1-p)(p^{|C^m|})+(1-p)(p^{|C^m|})^2&+\dots
\\
=(1-p)\sum_{t=0}^\infty (p^{|C^m|})^t&=\frac{1-p}{1-p^{|C^m|}}
\end{split}
\end{equation}
Secondly, each local graph is modeled as a random graph
$\mathcal{G}_u(N,\pi_u)$. Each node is expected to miss $rp_{eff}$
packets, therefore the average total number of missing (equivalent
to the average number of nodes in $\mathcal{G}$ and
$\mathcal{G}_u$'s), is $N=nrp_{eff}$. The total number of required
transmissions for the cooperative scenario is approximated by the
following theorem.
\begin{theorem}
\label{thm2}
 The total number of required transmissions for the
cooperative scenario $T_c$ is approximated by:
\begin{equation}
\begin{split}
T_c&=\mathrm{max}~t\\
s.t.\\
N_{t+1}&=N_t-2(1-p')\frac{\log N_t}{\log(\frac{1}{\pi_u})}\\
N_t&>0\\
\hat{p_{eff}}(t+1)&=\hat{p_{eff}}(t)-\frac{N_{t+1}-N_t}{nr}\\
N_0=N,
\pi_u&=(\frac{(|C^m|-2)(|C^m|-1)}{|C^m|^2}[(1-\hat{p}_{eff}(t))^4+(1-\hat{p}_{eff}(t))\hat{p}_{eff}^2]
\end{split}
\end{equation}
\end{theorem}
\textbf{where $N_t$ is the size of the remaining number of vertices
found by the algorithm at the end of $t$'th transmission and
$\hat{p_{eff}}(t)$ approximates the probability of missing a packet
(or equivalently the probability of a packet to be in the lacks set
of a client) at round $t$, where $\hat{p}(0)$ is set to be $p$
before the algorithm starts.}
\begin{proof}
At each round of the algorithm, each node tries to find the largest
clique over its local graph $G_u$. Cliques obtained from each local
graph are compared and the largest is chosen for removal. According
to Eq. \ref{myequa} and considering the parameters $\pi_u$ and $p'$
for the cooperative scenario, \textbf{it is expected that
$2(1-p')\frac{\log N_t}{\log(\frac{1}{\pi_u})}$ vertices are removed
at round $t$}. This process will continue until all the nodes are
removed (The recursive relation continues while $N_t>0$).
\textbf{$\hat{p}_{eff}(t)$ should be updated after each round of
transmission where similar discussion to $\hat{p}(t)$ in Section
\ref{sec:BaseStation} is applied. This discussion is removed due to
space limitations. $\pi_u$ is updated accordingly at each round $t$}

The probability $\pi_u$ that two nodes are connected to each other
in $G_u$ is calculated as follows:
\begin{equation}
\begin{split}
\pi_u= P(v_{ij}, v_{kl}\in E_u)=\\
P[((p_j\in \Gamma_k) ,(p_l\in \Gamma_i),(p_j,p_l)\in
\Gamma_u),(i\neq k\neq u)]\\
 +P[(((p_j=p_l)\notin (\Gamma_i\cup \Gamma_k))\\
  ,((p_j=p_l)\in \Gamma_u),(i\neq k\neq u)]=\\
  P[(p_j\in \Gamma_k)P(p_l\in \Gamma_i)P(p_j\in \Gamma_u)(p_l\in
  \Gamma_u)+\\
  P(p_j\notin \Gamma_i)P(p_j\notin \Gamma_k)P(p_j\in
  \Gamma_u)]P(i\neq k \neq u)=\\
  [(1-p_{eff})^4+(1-p_{eff})p_{eff}^2]\frac{(|C^m|-2)(|C^m|-1)}{|C^m|^2}
\end{split}
\end{equation}
\end{proof}
\subsection{Numerical Experiments}
Similarly to the numerical experiments for the case of centralized
packet recovery in Section \ref{sec:Basestation}, we have studied
two issues for the cooperative scenario: (a) the effect of
parameters $p$ and $r$ on the performance of the entire system and
(b) a comparison between simulations and the result of theorem
\ref{thm2}. However, it should be noted that in the cooperative
scenario other parameters such as the probability of erasure over
the link between two clients and also the size of a cluster are
incorporated.

Also, we should take this important fact into account that the data
rates for short range links between the clients are different from
the base station downlink data rates and are expected to be
typically larger which is an advantage for cooperation. Because of
the mentioned difference between the data rates of the short range
links and the long range base station downlink, it might not be
meaningful to count the total number of packets transmitted until
every clients would be able to decode, but we should separate the
base station broadcast phase from the cooperative packet recovery
phase where the clients use their short range links.

We have considered a set of $n=20$ clients and an arbitrary cluster
of size $|C^1|=8$ has been chosen. We assumed the probability of
erasure for the short range links is relatively small and is set to
be $p'=0.05$ for all the clients and is fixed during the recovery
phase. We have run a set of $1000$ experiments for each value of $p$
which takes values $0.1, 0.2, 0.3, 0.4$ and $0.5$. Fig.
\ref{fig:myfig3} shows the average total number of transmissions for
different values of $r$ ranging from $r=1$ to $r=20$. Clearly, the
total number of transmitted packets is expected to increase for
larger values of $r$ or $p$.

\begin{figure}[t]
\begin{center}
  \includegraphics{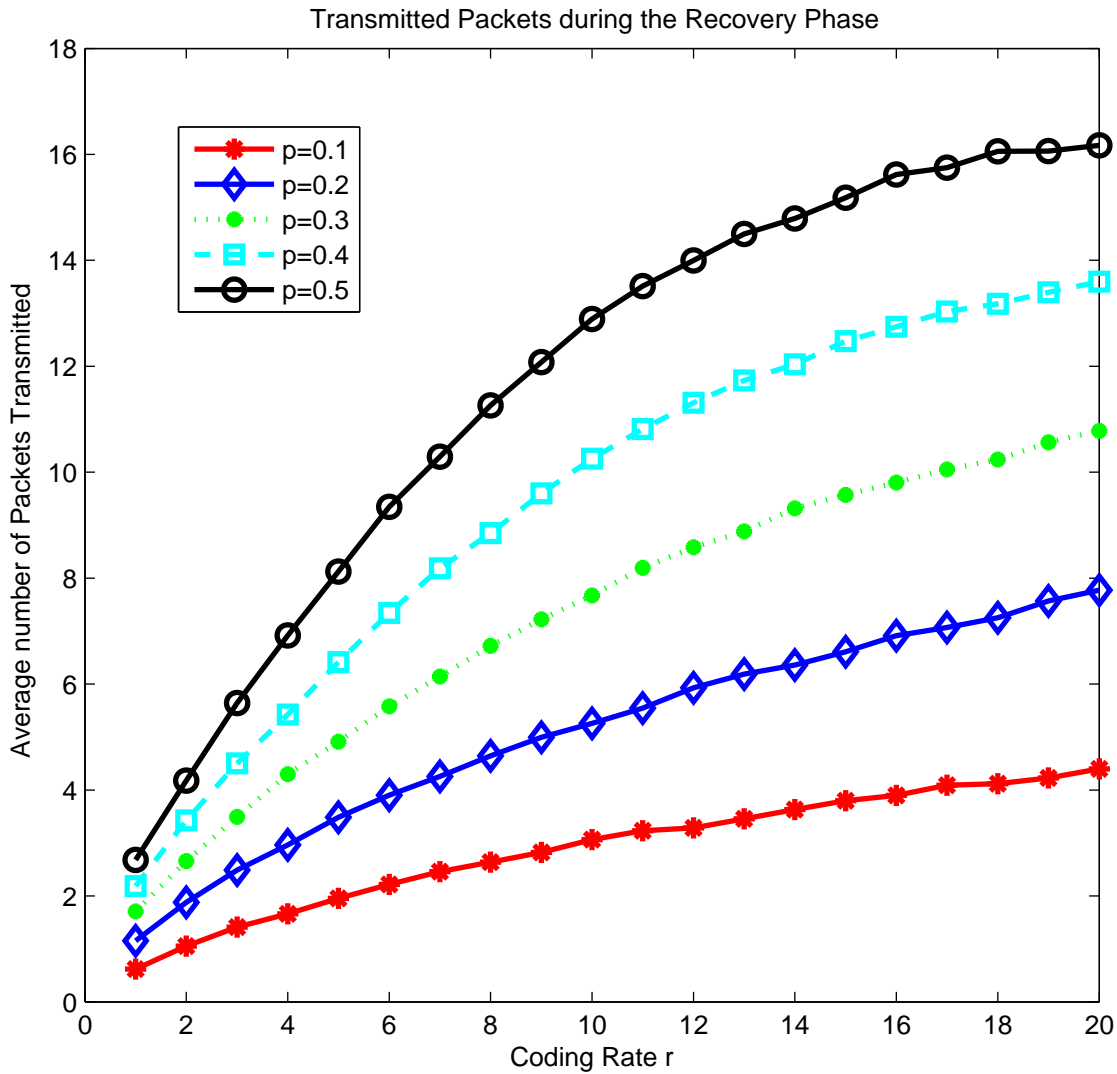}
  \caption{Total number of transmissions in the cooperative phase for different values of $r$ and $p$}\vspace{-1em} \label{fig:myfig3}
\end{center}
\end{figure}
%[trim = 30mm 85mm 10mm 85mm, clip, scale=.59]
Fig \ref{fig:myfig4} shows gain of network coding during the
cooperative recovery phase. Gain of network coding is defined as the
ratio of the total number of uncoded packets (denoted by $U_c$) i.e.
the number of packets required to be transmitted to satisfy the
demands of all the clients over a shared broadcast channel with
erasure probability $p'$ in the cooperative phase to the number of
actual transmitted packets during the cooperative recovery (denoted
by $Tc$). Therefore the gain of network coding would be
$\frac{U_c}{T_c}$ where
\begin{equation}
U_c=\frac{|\bigcup_{i': c_{i'}\in C^m} \Omega_{i'}|}{1-p'}
\end{equation}
 Fig. \ref{fig:myfig3} shows the average number of transmitted packets
 versus parameter $r$ for different values of $p$. Fig.
 \ref{fig:myfig4} shows the gain of network coding over uncoded
 transmissions. Gain of network coding implies how large the number
 of satisfied clients at each round rounds of transmission (those ones who can decode one packet
 at that round) is, which is a function of both $p$ and $r$.
 However,  as Fig. \ref{fig:myfig4} indicates the gain of network
 coding tends to be small if both the parameters are large, as all
 the clients have missed a large number of packets and the amount of
 side information at each client would be small. Therefore the
 opportunities of network coding tends to be less.
\begin{figure}[t]
\begin{center}
  \includegraphics [scale=0.5]{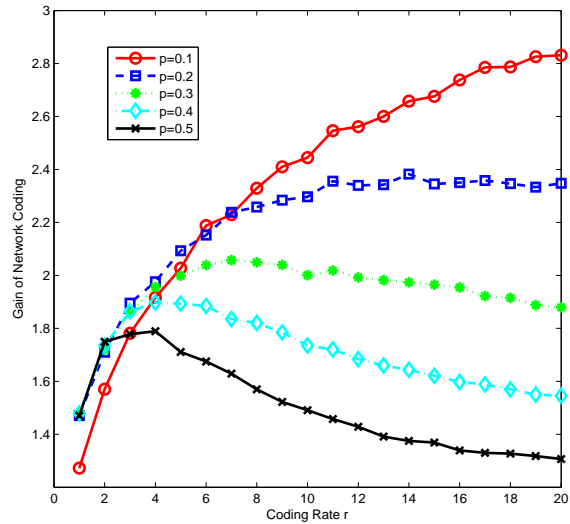}
  \caption{Gain of network coding in the cooperative phase}\vspace{-1em} \label{fig:myfig4}
\end{center}
\end{figure}
%[trim = 25mm 85mm 10mm 85mm, clip, scale=.59]
The other issue investigated in our numerical experiments is to
verify the result of theorem \ref{thm2}. We have considered $n=50$
clients and $r=40$. The probabilities of erasure are set to be
$p=0.3$ and $p'=0.1$. The number of removed nodes in each
transmission is compared for the simulations and as predicted in
Theorem \ref{thm2} in Fig. \ref{fig:myfig5}.
\begin{figure}[h]
\begin{center}
  \includegraphics [scale=0.5]{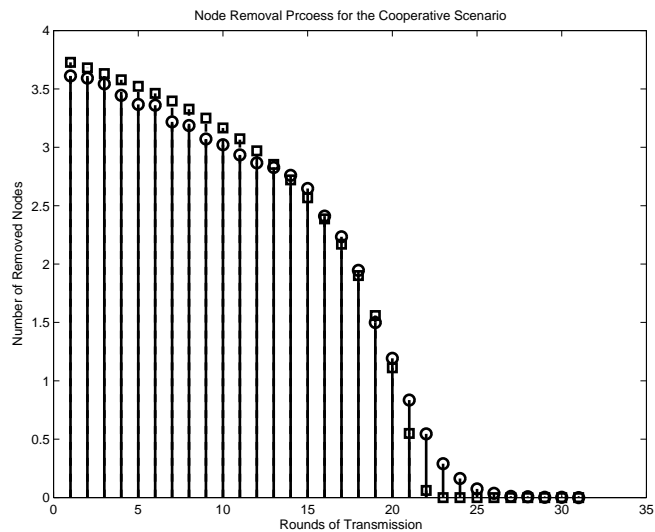}
  \caption{Theory vs. simulations for the packet recovery via cooperation.}\vspace{-1em} \label{fig:myfig5}
\end{center}
\end{figure}
%[trim = 20mm 75mm 10mm 65mm, clip, scale=.58]

%\input{Security}
\section{Conclusion} \label{sec:Conclusion}
We introduced a new method for combining unicasts to provide a
secure transmission method over a shared broadcast packet erasure
channel in a way that the secrecy of individuals is preserved. We
considered two different scenarios where in the first scenario the
base station should accomplish the transmission process until
everyone would be able to decode but in the second one the clients
were separated from the base station after an initial round of
transmission and cooperate with each other to recover the missing
packets. The amount of transmissions required for each scenario has
been evaluated theoretically and experimentally. A trade-off between
the total number of transmissions and secrecy against statistical
brute-force attacks can be observed. A more general information
theoretic analysis on the secrecy of the proposed method is an
interesting open problem and is not straightforward.

\IEEEpeerreviewmaketitle

\section*{Acknowledgment}

This work was supported under Australian Research Council Discovery
Projects funding scheme (project no. DP0984950 and project no.
DP120100160).

\bibliographystyle{IEEEtran}
%\bibliography{IEEEabrv,refs2009}

\end{document}